\documentclass[onecolumn, journal]{IEEEtran}

\usepackage{graphicx}
\usepackage{subfigure}
\usepackage{hyperref}
\usepackage{bbm}
\usepackage{amsmath}
\usepackage{amsthm}
\usepackage{amssymb}
\newtheorem{theorem}{\textbf{Theorem}}
\newtheorem{corollary}{\textbf{Corollary}}
\newtheorem{lemma}{\textbf{Lemma}}
\newtheorem{remark}{\textbf{Remark}}


\begin{document}

\title{An Importance Aware Weighted Coding Theorem Using Message Importance Measure\thanks{Zheqi Zhu, Shanyun Liu, Rui She, Shuo Wan, and Pingyi Fan are with the Department of Electronic Engineering, Tsinghua University, Beijing 100084, China (e-mail: fpy@tsinghua.edu.cn).
Khaled B. Letaief is with the Department of ECE, Hong Kong University of Science and Technology, Hong Kong.}}

\author{{Zheqi Zhu, Shanyun Liu, Rui She,  Shuo Wan, Pingyi Fan,~\IEEEmembership{Senior Member,~IEEE}, Khaled B. Letaief,~\IEEEmembership{Fellow,~IEEE}}}

% \IEEEauthorblockA{\IEEEauthorrefmark{1}
% Beijing National Research Center for Information Science and Technology and Department of Electronic Engineering Tsinghua University, Beijing 100084, China,
% }}
% \author{\IEEEauthorblockN{Author 1, Author 2, Author 3, Author 4, Author 5, Author 6}}

\maketitle

% \thispagestyle{empty} % no page number for the first page
% \pagestyle{empty}  % no page number for the second and the later pages

\begin{abstract}
There are numerous scenarios in source coding where not only the code length but the importance of each value should also be taken into account. Different from the traditional coding theorems, by adding the importance weights for the length of the codes, we define the average cost of the weighted codeword length as an importance-aware measure of the codes. This novel information theoretical measure generalizes the average codeword length by assigning importance weights for each symbol according to users' concerns through focusing on user's selections. With such definitions, coding theorems of the bounds are derived and the outcomes are shown to be extensions of traditional coding theorems.
\end{abstract}


\section{Introduction}
As one of the most fundamental theoretical basis of communication, source coding theorems aim at minimizing the average expected codeword length of the source symbols under the constraints of decipherable coding schemes. However, while taking users' preference into consideration, each source symbol may also vary in terms of importance. Thus, from the perspective of information theory, an importance aware measure for codes is necessary.

Let $X\in\mathcal X$ denote the finite discrete source to be encoded, where $\mathcal X=\{x_1,x_2,\cdots,x_N\}$ is the set of all available values of the source. $p_i$ and $l_i$ are the probability and the codeword length of $x_i$ respectively. Shannon's coding theorem~\cite{shannon1948mathematical} showed that the expected length $\bar{L}=\sum\limits_{i=1}^Np_il_i$ is close to $X$'s entropy, $H(X)=-\sum p_i\log p_i$. The theorem can be derived by solving an optimization problem to minimize $\bar{L}$ subject to the constraints of the Kraft's ineaquality~\cite{kraft1949device}: 
\begin{equation}\label{kraft}
    \sum\limits_{i=1}^ND^{-l_i}\leq1
\end{equation}
where $D$ is the size of the code alphabet~\cite{cover2012elements}. 

Based on the properties of information measure, Belis \textit{et al}.~\cite{belis1968quantitative} defined the function
\begin{equation}\label{Hu}
    H(U,X)=-\mathcal K\sum\limits_{i=1}^N u_ip_i\log p_i
\end{equation}
to measure the expected average of information utility, where $u_i$ is the utility of each event $x_i$ and $\mathcal K$ is the normalization factor. Inspired by this measure,~\cite{bhat2016noiseless} and~\cite{taneja1985coding} proposed the quantity $
    L_u=\frac{\sum u_ip_il_i}{\sum u_ip_i}$ 
as the useful expected codeword length and derived the R\'enyi-entropy related bounds under a so-called generalized Kraft's ineaquality
\begin{equation}\label{gen kraft}
    \sum\limits_{i=1}^Nu_iD^{-l_i}\leq\sum\limits_{i=1}^Nu_ip_i.
\end{equation} Considering the 1-order results, they obtained the generalized useful information source coding (UISC) bounds:
\begin{equation}\label{u coding}
    \frac{-\sum u_ip_i\log_Dp_i}{\sum u_ip_i}\leq L_u<\frac{-\sum u_ip_i\log_Dp_i}{\sum u_ip_i}+1.
\end{equation}
%However, the counterexamples in the following section will show that such an ineaquality does not always hold for uniquely decodable prefix codes which means that it cannot be a constraint in such cases.

Motivated by above ideas, we shall consider the coding scenarios where symbols have various weights which can reflect the impact of human's opinion on each value. Thus, the length should also be weighted. Specifically, message importance measure (MIM) introduced in~\cite{fan2016message} gives an exponential-form measure to quantify the importance of the event from a probabilistic viewpoint. MIM has been used in several fields such as information compression, distribution estimation, anomaly detection, recommendation systems and IoTs~\cite{liu2018non,she2019importance,liu2019matching}. The parameter in MIM can be selected to amplify the weights of certain elements according to users' preference~\cite{she2017focusing}. Thus, by assigning MIM weightings to code lengths, an importance-aware codeword measure with explicit forms can be defined.
% we define the importance-aware code length. Also, the coding theorem and several properties will be discussed.

The main contributions of this work can be summarized as follows: (1) We propose an importance-aware weighted expected code length and formulate the optimization problem, as well as derive lower/upper bounds for single sources and sequence sources respectively, which cover Shannon's coding theorem as a special case; 
(2) The MIM-form weighted coding theorems and their properties are discussed, which results in a focusing effect to the events with certain probability; (3) Numeric simulations and comparisons with relative coding theorems are presented to investigate the characteristics of the theorems. The results will show better convergence, length-stability and preference reflection of MIM based I-W coding.

The rest of this letter is organized as follows. In Section \uppercase\expandafter{\romannumeral2} we will define the importance-aware code length and derive its coding theorems for two types of sources. In Section \uppercase\expandafter{\romannumeral3}, we discuss the theorems under specific weightings, namely, the MIM-weighted measure. Then, in Section \uppercase\expandafter{\romannumeral4}, some simulation results will be shown to demonstrate the properties and the advantages of such information theoretical measure and its coding theorems. Finally, in Section \uppercase\expandafter{\romannumeral5} we conclude this work and give several possible research directions. 
%In addition, some basic notations used in the following sections are listed in Tabel~\ref{notation}.
% \begin{table}\label{notation}
%     \centering
%     \begin{tabular}{ccc}
%         \hline
%         Notation & Description\\
%         \hline
%         $X$ & The source random variable to be encoded\\
%         $\mathcal X$ & The set of possible values of $X$\\
%         $p_i$ & The probability of value $x_i\in\mathcal X$, i.e $P(x_i)$\\
%         $l_i$ & The codeword length of $x_i$\\
%         $D$ & The size of encoding alphabet\\
%         \hline
%     \end{tabular}
%     \caption{The description of notations.}
% \end{table}

\section{An Importance-aware Weighted Measure and Its Coding Theorem}
Motivated by Eq.(\ref{Hu}), we take importance weight of each symbol in $\mathcal X$ into consideration and define the quantity
\begin{equation}
    \label{Lw}
    \bar L_w:=\sum\limits_{i=1}^Np_iw_il_i
\end{equation}
as the \textbf{i}mportance-aware \textbf{w}eighted (I-W) expected code length, where $w_i$ is the importance weight assigned for $x_i$. 
$\bar{L}_w$ can be regarded as a cost of certain encoding schemes and the corresponding coding theorems can be derived by minimizing the cost under the encoding constraints. However, the generalization of the Kraft's inequality, Eq.(\ref{gen kraft}) under the weighted length definitions proposed in~\cite{bhat2016noiseless} and~\cite{taneja1985coding} does not always hold for uniquely decodable prefix codes. Assume that we have a Bernoulli source $X$ which has two possible taking values $\mathcal X=\{x_1, x_2\}\sim\{p_1,p_2=1-p_1\}$ and set $\{u_i\}=\{1,2\}$. It is obvious that its uniquely decodable prefix binary codes can be 0 and 1. Thus, the left side of Eq.(\ref{gen kraft}) is constantly equal to $\frac{1}{2}(u_1+u_2)=\frac{3}{2}$ while the right side is lower than the left if $p_1>0.5$. 
%Moreover, if we select the MIM-formed utility, $u_i=e^{\omega(1-p_i)}$, the result is completely opposite. 
This counterexample shows that the so-called generalized Kraft's inequality Eq.(\ref{gen kraft}) is not necessary for uniquely decodable codes and the constraint for such I-W weighted length should still be the original Kraft's inequality, Eq.(\ref{kraft}).
%Then, we shall derive the coding theorems based on I-W weighted code length for single sources as well as sequence cases respectively.
\subsection{Coding Theorem for Single Source}
For a single source, the aim is to find the code with lengths $l_1,\cdots,l_N$ for each taking value which satisfy the Kraft's inequality Eq.(~\ref{kraft}) and lead to the minimum expected I-W length $\bar L_w$. Then, we have the following optimization problem:
\begin{align}
    \mathcal P_1:\quad &\min\limits_{l_1,\cdots,l_N}\ \sum\limits_{i=1}^Np_iw_il_i\\
    &\begin{array}{l@{\ }l@{}l@{\ }l}
        \notag
        \mbox{s.t.}&\sum\limits_{i=1}^ND^{-l_i}&\leq1.\\
    \end{array}
\end{align}
By solving the above optimization problem, we obtain the upper and the lower bounds of the optimal I-W expected codeword length. This leads to the following importance-aware coding theorem for single source.
\begin{theorem}[I-W Coding Theorem for Single Source]\label{THM1}
    Let $l^*_1,\cdots,l^*_N$ be the optimal codeword lengths for a single source with the distribution $\{p_i\}$ and $D$ be the size of the coding alphabet. Then, the associated I-W expected codeword length $\bar L_w$ of the optimal code satisfies:
    \begin{equation}\label{thm1}
        \mathcal L(w,X)\leq\bar L^*_w<\mathcal L(w,X)+H_w(X)
    \end{equation}
    where $H_w(X)=\sum\limits_{j=1}^Np_jw_j$ is the probabilistic average of weightings and $\mathcal L(w,X)=-\sum\limits_{i=1}^Np_iw_i\log_D\frac{p_iw_i}{H_w(X)}$ is called the importance-aware measure based on weightings $\{w_i\}$ and random variable $X$.
\end{theorem}
\begin{proof}
    To solve Problem $\mathcal{P}_1$, we consider the Lagrangian function
    \begin{equation}
        J(\boldsymbol l,\lambda)=\sum\limits_{i=1}^Np_iw_il_i+\lambda\left(\sum\limits_{i=1}^ND^{-l_i}-1\right)
    \end{equation}
    The function is obviously convex. Firstly, without regards to the integer constraint of $l_i$, calculating the partial derivaties with respect to $\{l_i\}$,\ $\lambda$ and setting them to $0$, we obtain 
    \begin{align}
        \label{2.9}
        0&=\frac{\partial J}{\partial l_i}\bigg|_{\tilde l^*_i}=p_iw_i-\lambda D^{-l_i}\ln D\bigg|_{\tilde l^*_i},\ i=1,\cdots,N\\
        \label{2.10}
        0&=\frac{\partial J}{\partial\lambda}\bigg|_{\lambda^*}=\sum\limits_{i=1}^ND^{-l_i}-1.
    \end{align}
    Eq.(\ref{2.9}) leads to $D^{-\tilde l^*_i}=\frac{p_iw_i}{\lambda\ln D}$. Then, by substituting the term into Eq.(\ref{2.10}), we obtain
    \begin{equation}
        \tilde l^*_i=-\log_D\frac{p_iw_i}{\sum p_iw_i}=-\log_D\frac{p_iw_i}{H_w}
    \end{equation}
    and
    \begin{equation}
        \tilde{\bar L}^*_w=\mathcal L(w,X).
    \end{equation}
    Moreover, the codeword length $l_i$ should be integer which means that $l^*_i=\lceil\tilde{l}^*_i\rceil$, i.e., 
    \begin{equation}\label{li}
        -\log_D\frac{p_iw_i}{H_w}\leq l^*_i<-\log_D\frac{p_iw_i}{H_w}+1.
    \end{equation}
    Hence, we obtain the Eq.(\ref{thm1}). 
\end{proof}
\begin{remark}\label{rm1}
    Note that if we set $w_i=1$ for all symbols, the theorem is exactly equivalent to Shannon's coding theorem for single symbols: $H(X)\leq\bar L<H(X)+1.$
\end{remark}

\subsection{Coding Theorem for Sequence Source}
In the cases where the sequences of $n$ symbols are encoded, we regard the sequence $\boldsymbol X=(X_1,\cdots,X_n)$ as a supersymbol. Similarly, we define the I-W expected codeword length per symbol as 
\begin{equation}\label{lnw}
    \bar L_{n,w}:=\frac{1}{n}\sum\limits_{\boldsymbol x=x_1,\cdots,x_n}p(\boldsymbol x)w(\boldsymbol x)l(\boldsymbol x)
\end{equation}
where $w(\boldsymbol x)$ is the weight of an $n$-length sequence. Especially, while $X_1,\cdots,X_n$ are i.i.d., the weight is the product of the weights of all the sub-symbols,
\begin{equation}\label{wx}
    w(\boldsymbol x)=w(x_{k_1},\cdots,x_{k_n})=\prod\limits_{i=1}^nw_{k_i}
\end{equation}
Then, we can obtain the I-W coding theorem for sequences as \textbf{Theorem}~\ref{THM2}.
\begin{theorem}[I-W Coding Theorem for Sequence Source]\label{THM2}
    For a sequence with $n$ symbols, the optimal I-W expected codeword length per symbol $\bar L^*_{n,w}$ satisfies 
    \begin{equation}\label{thm2ori}
        \frac{\mathcal L(w,\boldsymbol X)}{n}\leq\bar L^*_{n,w}<\frac{\mathcal L(w,\boldsymbol X)}{n}+\frac{H_w(\boldsymbol X)}{n}
    \end{equation}
    where
    \begin{equation}\label{hwx}
        H_w(\boldsymbol X)=\sum\limits_{\boldsymbol x}p(\boldsymbol x)w(\boldsymbol x)
    \end{equation}
    is the probabilistic average of sequence weightings and 
    \begin{equation}\label{lwx}
        \mathcal L(w,\boldsymbol X)=-\sum\limits_{\boldsymbol x}p(\boldsymbol x)w(\boldsymbol x)\log_D\frac{p(\boldsymbol x)w(\boldsymbol x)}{H_w(\boldsymbol X)}.
    \end{equation}
    Moreover, if all the sub-symbols $X_1,\cdots,X_n$ are i.i.d., the bounds of $\bar L^*_{n,w}$ can be written as 
    \begin{equation}\label{thm2simple}
        H_\omega^{n-1}(X)\mathcal L(w,X)\leq\bar L^*_{n,w}<H_\omega^{n-1}(X)\mathcal L(w,X)+\frac{H_\omega^{n}(X)}{n}
    \end{equation}
\end{theorem}
\begin{proof}
    Eq.(\ref{thm2ori}) is the direct result of \textbf{Theorem}~\ref{THM1} for super symbol $\boldsymbol X=(X_1,\cdots,X_n)$. While $X_1,\cdots,X_n$ are i.i.d., Eq.(\ref{thm2simple}) can be proved by the following \textbf{Lemma}~\ref{lemma1}.
\end{proof}
\begin{lemma}\label{lemma1}
    Let $\boldsymbol X= (X_1,\cdots,X_n)$ be the super symbol, for i.i.d. cases, we have the following equalities, 
    \begin{align}
        H_w(\boldsymbol X)&=H_w^n(X)\\
        \mathcal L(w,\boldsymbol X)&=nH_\omega^{n-1}(X)\mathcal L(w,X)
    \end{align}
\end{lemma}
\begin{proof}
    Firstly, for i.i.d. cases where $p(\boldsymbol x)=\prod\limits_{i=1}^np(X_i)$, by Eq.(\ref{wx} and Eq.(\ref{hwx}), we have
    \begin{equation}\notag
        \begin{aligned}
            H_w(\boldsymbol X)&=\sum\limits_{\boldsymbol x=X_1,\cdots,X_n}\prod\limits_{i=1}^np(X_i)w(X_i)\\
            &=\prod\limits_{i=1}^n\left(\sum\limits_{X_i}p(X_i)w(X_i)\right)=H_w^n(X).
        \end{aligned}
    \end{equation} 
    Then, from Eq.(\ref{lwx}), we have
    \begin{equation}\notag
        \begin{aligned}
            \mathcal L(w,\boldsymbol X)&=
            % -\sum\limits_{\boldsymbol x}p(\boldsymbol x)w(\boldsymbol x)\log_D\frac{p(\boldsymbol x)w(\boldsymbol x)}{H_w(\boldsymbol X)}\\
            -\sum\limits_{\boldsymbol x}\prod\limits_{i=1}^np(X_i)w(X_i)\log_D\frac{\prod\limits_{j=1}^np(X_j)w(X_j)}{H_w^n(X)}\\
            &=-\sum\limits_{\boldsymbol x}\sum\limits_{j=1}^n\log_D\frac{p(X_j)w(X_j)}{H_w(X)}\prod\limits_{i=1}^np(X_i)w(X_i)\\
            &=nH_\omega^{n-1}(X)\mathcal L(w,X).
        \end{aligned}
    \end{equation}
\end{proof}
\begin{remark}
    By setting $w_i=1$, we obtain Shannon's coding theorem for sequences:
    \begin{equation}
        \frac{H(X_1,\cdots,X_n)}{n}\leq\bar L_n<\frac{H(X_1,\cdots,X_n)}{n}+\frac{1}{n}
    \end{equation}
    and $H(X)\leq\bar L_n<H(X)+\frac{1}{n}$ for i.i.d. cases.
\end{remark}

\section{Coding Theorem under MIM-weighted Measure}
We shall now consider a specific form of the importance weightings. Message importance measure (MIM)~\cite{fan2016message}, as an information theoretical measure of the importance, is defined from the probability of each taking value as
\begin{equation}\label{mim}
    M\!I\!M(X;\omega)=\sum\limits_{i=1}^Np_ie^{\omega(1-p_i)}
\end{equation}
where $\omega$ is the importance coefficient. MIM has similar properties to Shannon's entropy such as convexity property, independent probability property, and minimum/maximum value property. Besides, considering the importance proportion of each value, the unnormalized importance factor $M\!I\!M(x_i;\omega)$ of taking value $x_i$ is defined as $M\!I\!M(x_i;\omega)=p_ie^{\omega(1-p_i)}$. Note that $M\!I\!M(X;\omega)$ is actually the sum of $M\!I\!M(x_i;\omega)$ for all taking values. Then, dividing by $M\!I\!M(X;\omega)$, the normalized form of the importance for each value is
\begin{equation}
    \label{nmim}
    M\!I\!M\!_N(x_i;\omega)=\frac{M\!I\!M(x_i;\omega)}{M\!I\!M(X;\omega)}.
\end{equation}
% \begin{equation}
%     \label{nmim}
%     M\!I\!M(x_i;\omega)=\frac{p_ie^{\omega(1-p_i)}}{\sum\limits_{j=1}^Np_je^{\omega(1-p_j)}}.
% \end{equation}
MIM can zoom in small or large probability by setting different parameters $\omega$. Generally, a positive $\omega$ leads to more focus on smaller probability elements and a negative $\omega$ amplifies the impact of larger probability elements. Moreover, it is noted that $x_i$ has the most contribution if $\omega$ is set to be $\frac{1}{p_i}$, which means that the parameter can be properly selected to reflect users' preference. For instance, while users consider the importance of each class from the probabilistic view, by assigning $\omega$ according to the class they prefer, such elements take the most weights~\cite{she2017focusing,liu2018switch}.  Motivated by this, we set MIM as the importance weights in Eq.(\ref{Lw}),
\begin{equation}
    \label{miml}
    \bar L(\omega):=\sum\limits_{i=1}^NM\!I\!M\!_N(x_i;\omega)l_i
\end{equation}
and obtain the codeword length bounds under the MIM-weighted measure as the following corollaries.

\begin{corollary}[MIM-weighted Coding Theorem for Single Source]\label{cor1}
    For MIM-weighted codeword length with the parameter $\omega$, the optimal measure $\bar L^*(\omega)$ satisfies:
    \begin{equation}
        H\!_{M\!I\!M}(X;\omega)\leq\bar L^*(\omega)<H\!_{M\!I\!M}(X;\omega)+1
    \end{equation}
    where
    \begin{equation}\notag
        H_{M\!I\!M}(X;\omega)=-\sum\limits_{i=1}^NM\!I\!M\!_N(x_i;\omega)\log_DM\!I\!M\!_N(x_i;\omega)
    \end{equation}
    can be regarded as the entropy under the MIM measure.
\end{corollary}
\begin{corollary}[MIM-weighted Coding Theorem for Sequence Source]\label{cor2}
    The optimal MIM-weighted codeword length $\bar L^*_n(\omega)$ for i.i.d. sequence source is bounded by:
    \begin{equation}
            H\!_{M\!I\!M}(X;\omega)\leq\bar L^*_n(\omega)<H\!_{M\!I\!M}(X;\omega)+\frac{1}{n}
    \end{equation}
\end{corollary}
\begin{proof}
    Set the I-W factors in Eq.(\ref{Lw}) and Eq.(\ref{wx}) as the normalized exponential weights, $w_i=\frac{e^{\omega(1-p_i)}}{M\!I\!M(X;\omega)}$, by \textbf{Theorem}~\ref{THM1} and \textbf{Theorem}~\ref{THM2} the corollaries above can be proved.
\end{proof}
\begin{remark}
    \label{rm3}
    Shannon's coding bounds are special cases of the MIM-weighted coding theorems when $\omega=0$.
\end{remark}

\section{Numerical Results}
In this section, we will show several simulation results of the MIM weighted coding theorems and investigate the properties of the I-W measure as well as its codeword bounds.
\begin{figure}[htbp]
    \vspace{-0.02\textheight}
    \setlength{\abovecaptionskip}{-0.007\textheight}
    \setlength{\belowcaptionskip}{-0.008\textheight}
    \centering
    \subfigure[$\omega=-1$.]{
      \label{f01} %% label for first subfigure
      \begin{minipage}{0.48\linewidth}
        \includegraphics[width=1\textwidth]{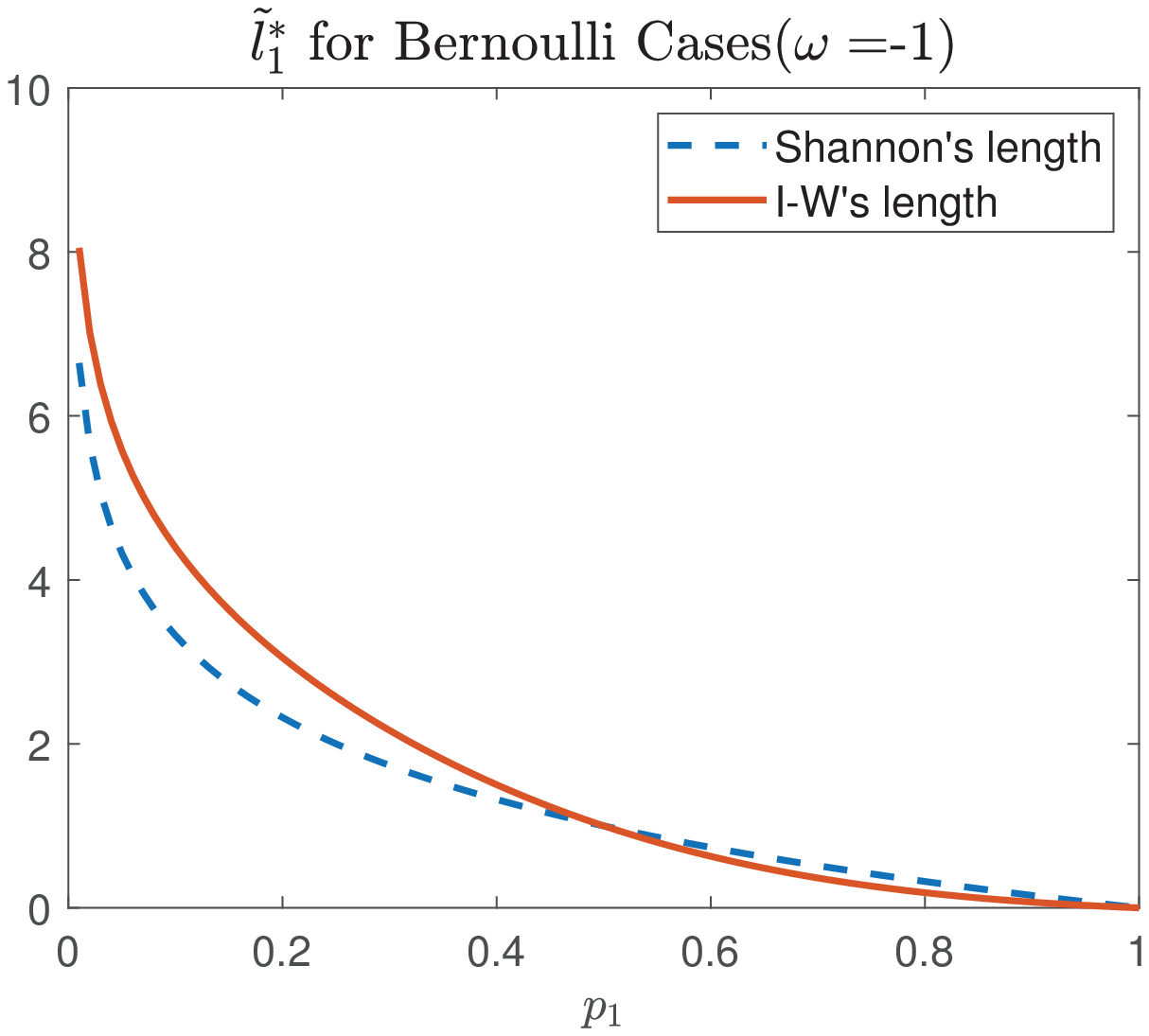}
      \end{minipage}}%
    \subfigure[$\omega=1$.]{
      \label{f02} %% label for second subfigure
      \begin{minipage}{0.48\linewidth}
        \includegraphics[width=1\textwidth]{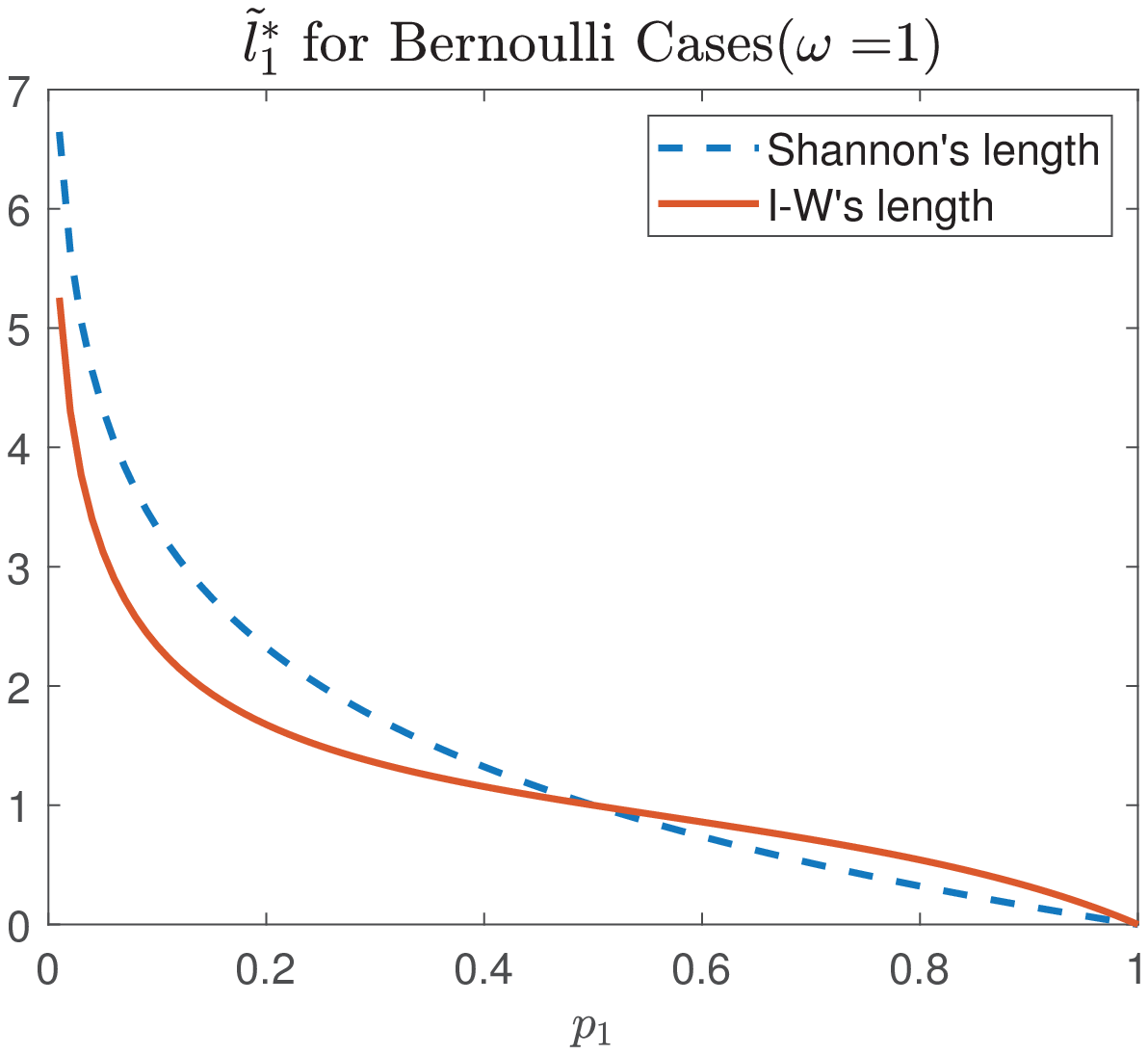}
      \end{minipage}}
    \caption{The optimal length $\tilde l^*_1$ for binary sources.}
    \label{fig0} %% label for entire figure
    \vspace{-0.025\textheight}
\end{figure}

Fig.\ref{fig0} shows the optimal codeword length $\tilde l^*_1$ for $x_1$ with various $p_1$ for Bernoulli sources under Shannon's coding theorem and I-W's according to Eq.(\ref{li}). It is obvious that $\tilde l^*_1$ under the MIM weighted measure shows different features compared to Shannon's coding theory. Specifically, for $\omega<0$, MIM weighted measure magnifies the influence of symbols with larger probability which results in larger importance and less redundency for codewords. On the contrary, while $\omega>0$, MIM weighted measure focuses on the symbols with smaller probability and hence the characteristics are opposite. In other words, the codeword length under MIM weighted coding theorems is more flexible than Shannon's and we can choose to make compression or reserve redundancy for the symbols users concern by setting the corresponding importance coefficient $\omega$.
\begin{figure}[htbp]
    \vspace{-0.02\textheight}
    \setlength{\abovecaptionskip}{-0.007\textheight}
    \setlength{\belowcaptionskip}{-0.008\textheight}
    \centering
    \subfigure[$\omega=-4$.]{
      \label{f1} %% label for first subfigure
      \begin{minipage}{0.48\linewidth}
        \includegraphics[width=1\textwidth]{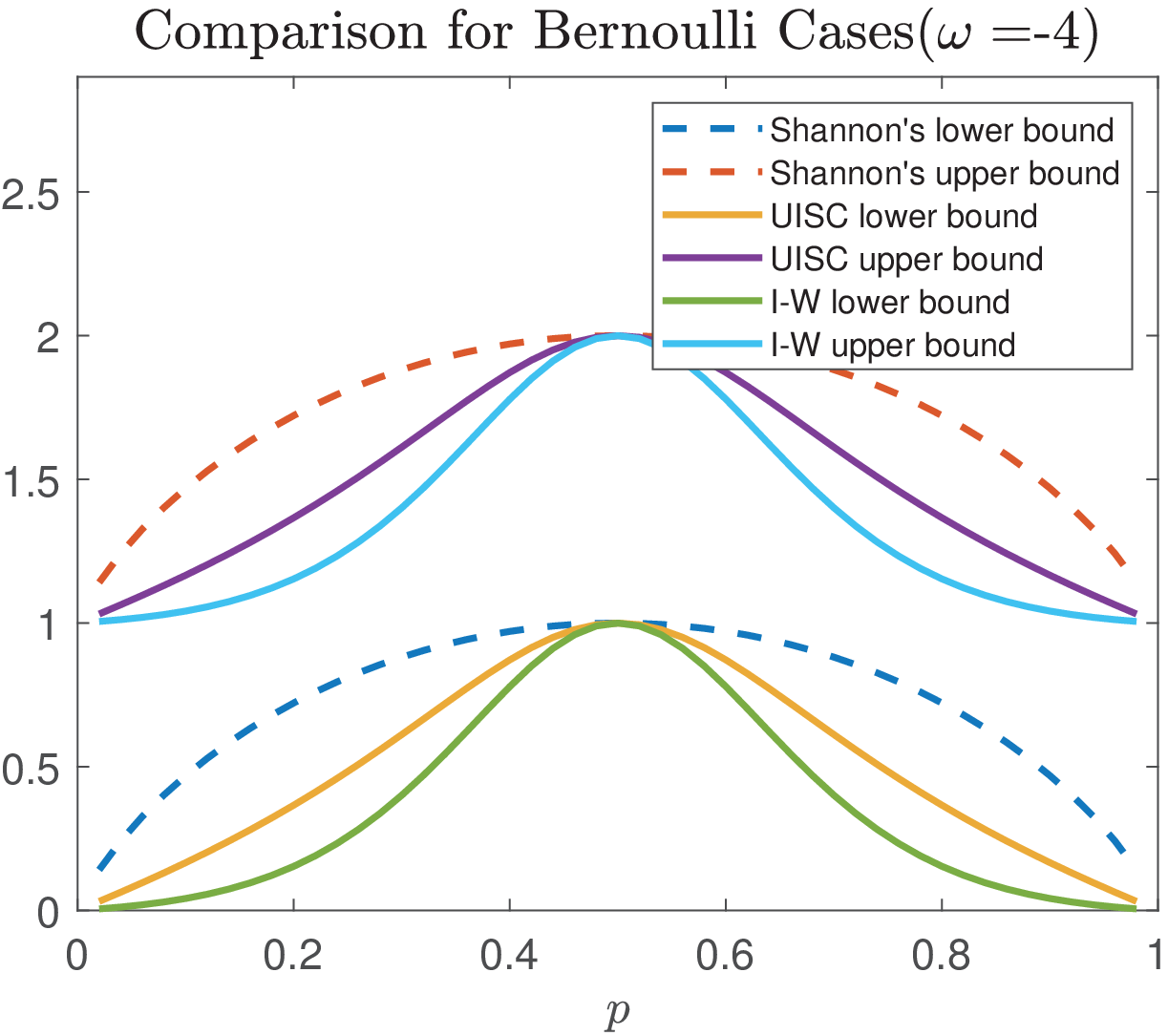}
      \end{minipage}}%
    \subfigure[$\omega=1$.]{
      \label{f2} %% label for second subfigure
      \begin{minipage}{0.48\linewidth}
        \includegraphics[width=1\textwidth]{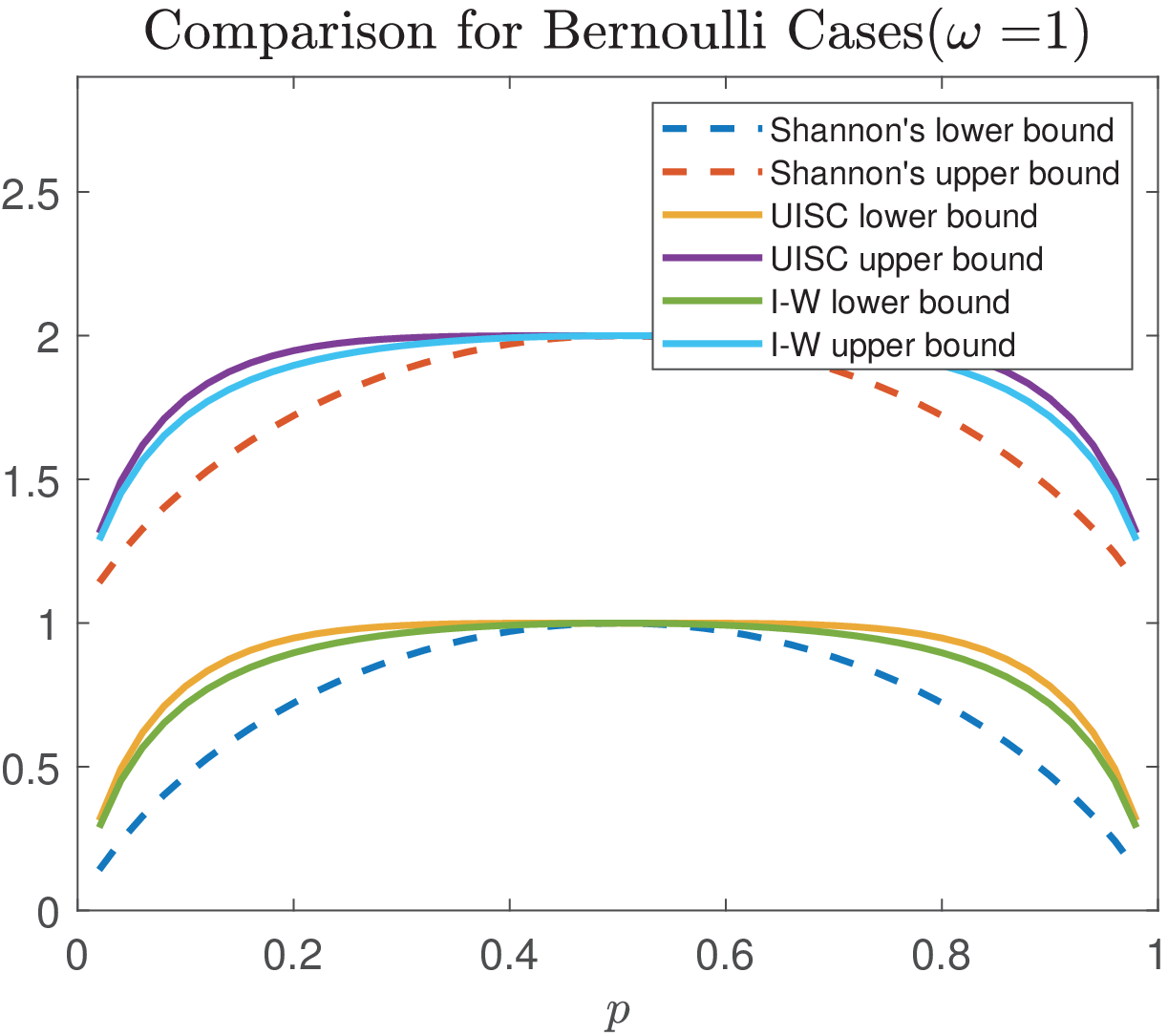}
      \end{minipage}}

    \subfigure[$\omega=4$.]{
      \label{f3} %% label for second subfigure
      \begin{minipage}{0.48\linewidth}
        \includegraphics[width=1\textwidth]{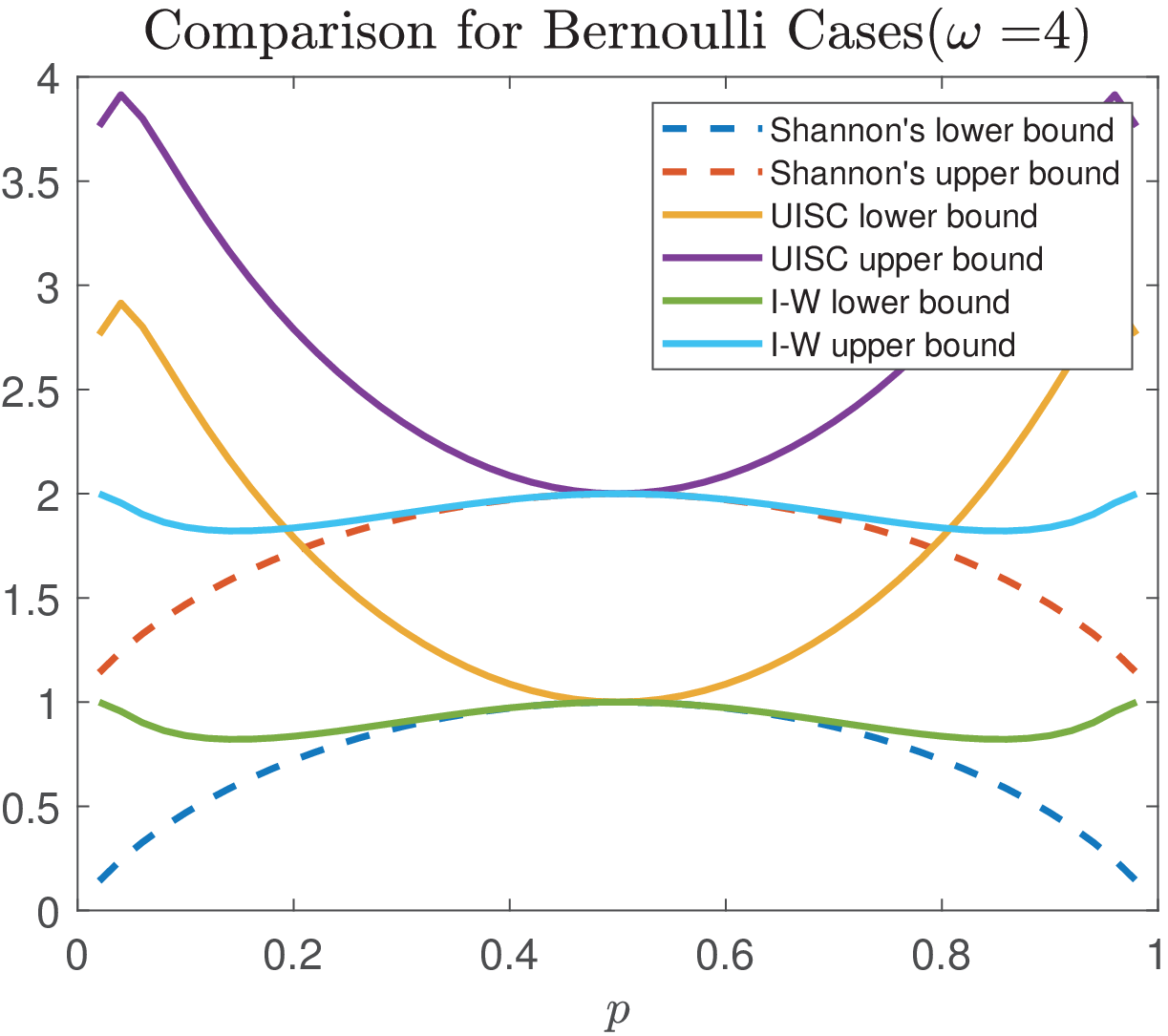}
      \end{minipage}}
    \subfigure[$\omega=8$.]{
      \label{f4} %% label for second subfigure
      \begin{minipage}{0.48\linewidth}
        \includegraphics[width=1\textwidth]{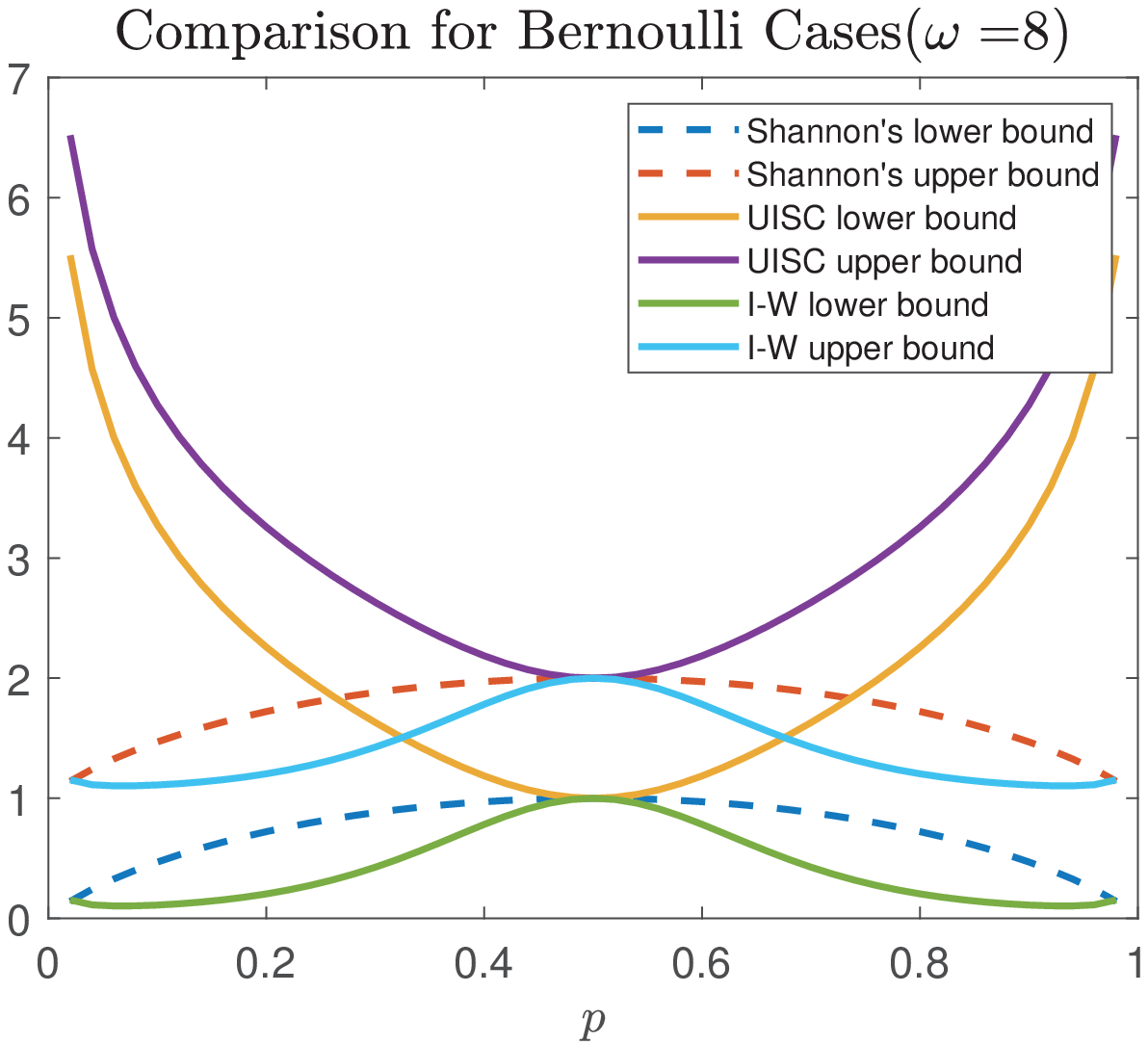}
      \end{minipage}}
    \caption{Comparisons of Shannon's coding theorem, UISC and I-W coding theorem for Bernoulli sources.}
    \label{fig1} %% label for entire figure
    \vspace{-0.015\textheight}
\end{figure}

In Fig.\ref{fig1}, under different MIM coefficient $\omega$, we compare the codeword length bounds of Shannon's, UISC~\cite{taneja1985coding} and proposed I-W coding for Bernoulli sources. 
%Firstly, by \textbf{Remark}~\ref{rm3}, the bounds under MIM weighted measure goes close to Shannon's when $\omega$ is close to 0, as shown in Fig.\ref{f1}. Then, for small $\omega$ in Fig.\ref{f1} to \ref{f3}, the MIM weighted measure bounds are larger than Shannon's for all $p$, which implies that the I-W coding theorems bring more redundancy than Shannon's coding schemes to fit the renewed probability density caused by the importance weightings. Such redundancy can be used to guarantee that the symbols with more importance can be encoded more robustly by increasing the code lengths. Besides, the results are interesting when importance coefficient gets larger. Fig.\ref{f4} shows that for larger $\omega$, the MIM weighted measure bounds are lower than Shannon's because the I-W theorem reduces the codeword lengths of symbols with smaller probabilities but higher importance by balancing the importance-aware density. 
Firstly, as shown in Fig.\ref{f1}, while $\omega$ is set to be negative, MIM increases the weights of symbols with larger probability which leads to compression for both weighted coding theorems. Then, for small positive $\omega$ in Fig.\ref{f2} and \ref{f3}, MIM weighted measure bounds are larger than Shannon's for all $p$, which implies that the I-W and UISC coding theorems bring more redundancy than Shannon's coding schemes to fit the renewed probability density caused by the importance weightings. Such redundancy preserves enough codeword length for the symbols of users' interests. Fig.\ref{f4} shows that for larger $\omega$, MIM based I-W bounds are lower than Shannon's because it reduces the codeword lengths of symbols with smaller probabilities but higher importance by balancing the importance-aware density. Furthermore, under MIM weightings with large $\omega$, MIM based I-W coding theorem outperforms UISC on the convergency and stability of the weighted codeword length, as demonstrated in Fig.\ref{f3} and \ref{f4}. 
According to the above properties, MIM weighted coding theorem can be used to guide the compression encoding through a trade-off between probabilities and importance measures. In particular, for uniform distributed sources, neither compression nor redundancy occurs in I-W coding theorems because all of the weights are equal.

\section{Conclusion and Future Work}
In this letter, we proposed an importance-aware weighted codeword length and derived its coding theorems. We also discussed the specific form of the weightings and obtained the version based on MIM. Numerical results identified some differences with Shannon's source coding theorems and showed advantages compared to other weighted coding theorems. 

The I-W coding theorems introduced the concept of importance weightings to the codeword length and can be used in compression coding and in scenarios where the distribution of the source is imbalanced. For further work, specific coding schemes and algorithms need to be investigated based on the bounds developed in this letter.

\nocite{*}
\bibliographystyle{unsrt}
\bibliography{refs}

\end{document}